\documentclass[nolineno,a4paper,UKenglish,cleveref, autoref, thm-restate]{socg-lipics-v2021}

\hideLIPIcs

\usepackage{standalone}
\usepackage{algorithm}
\usepackage{algpseudocode}
\usepackage{amsmath}
\usepackage{amsthm}
\usepackage{booktabs}
\usepackage{url}
\usepackage{tikz}
\usetikzlibrary{patterns}
\usetikzlibrary{positioning}
\usetikzlibrary{decorations.pathreplacing}

\algrenewcommand\alglinenumber[1]{\footnotesize #1}
\def\blockat{\textit{blockat}}
\def\key{\textit{key}}

\title{Parallel $\mathcal O(\sqrt{n})$ Overhead LSD Radix Sort}
\author{Robert Clausecker}{Zuse Institute Berlin, Germany}{clausecker@zib.de}{}{NHR-Verein Graduate School programme}
\author{Florian Schintke}{Zuse Institute Berlin, Germany}{schintke@zib.de}{https://orcid.org/0000-0003-4548-788X}{}
\authorrunning{R. Clausecker and F. Schintke}
\Copyright{Robert Clausecker and Florian Schintke} 

\ccsdesc[500]{Theory of computation~Sorting and searching}

\keywords{cache locality, radix sort, sorting} 

\supplement{https://github.com/clausecker/radsort}

\acknowledgements{This work was produced during a secondment of the first author at the
Sanders workgroup at KIT.  The author would like to thank Marvin Williams for his invaluable
input.}

\EventEditors{John Q. Open and Joan R. Access}
\EventNoEds{2}
\EventLongTitle{42nd Conference on Very Important Topics (CVIT 2016)}
\EventShortTitle{CVIT 2016}
\EventAcronym{CVIT}
\EventYear{2016}
\EventDate{December 24--27, 2016}
\EventLocation{Little Whinging, United Kingdom}
\EventLogo{}
\SeriesVolume{42}
\ArticleNo{23}

\begin{document}

\maketitle

\begin{abstract}
We present
Radsort, a variant of LSD radix sort, sorting data with $\mathcal O(\sqrt{n})$~additional space.
Radsort is stable, admits a simple implementation and is easy to parallelise.
For arrays exceeding a size of around~$2\,\mathrm{MiB}$ it outperforms a conventional
out-of-place LSD radix sort.
\end{abstract}

\section{Introduction}
An LSD (least significant digit) radix sort algorithm sorts an array~$A$ of
length~$n$ holding $n_t$-tuples of keys $\bigl(\key(A[i],0),\key(A[i],1),\ldots,\key(A[i],n_t-1)\bigr)$
of alphabet~$\Sigma$ with $\sigma$~symbols into lexicographic order
by repeatedly sorting according to $\key(A[i],n_t-1)$,
then~$\key(A[i],n_t-2)$, and so on until~$\key(A[i],0)$.
With a stable sort, identical keys end up sorted by their suffixes,
giving a lexicographic order.
A conventional LSD radix sort sorts by~$\key(A[i],t)$ by first taking a histogram~$H$
of how often each $c\in\Sigma$ occurs.  A prefix sum over~$H$ obtains
the bucket starts~$S$, with which we sort~$A$ into a new array~$A'$
(Alg.~\ref{alg:ooplsd}).

\begin{algorithm}
\caption{out-of-place LSD radix sort round}\label{alg:ooplsd}
\begin{algorithmic}[1]
\Procedure{radixSortStep}{$A',A,t$}\Comment{sort array~$A$ by key~$t$ into~$A'$}
\State $H\gets 0\ldots 0$
\For{$i\gets 0\ldots n-1$}\Comment{take a histogram of $\key(A[i],t)$}
 \State$H[\key(A[i],t)]\gets H[\key(A[i],t)]+1$
\EndFor
\State $S[0] \gets 0$
\For{$i\gets 1\ldots\sigma-1$}
 \State$S[i]\gets S[i-1]+H[i-1]$\Comment{prefix sum over $H$}
\EndFor
\For{$i\gets 0\ldots n-1$}
 \State$c\gets\key(A[i],t)$\Comment{determine bucket}
 \State$A'\bigl[S[c]\bigr]\gets A[i]$\Comment{sort $A[i]$ into $A'$}
 \State$S[c]\gets S[c]+1$\Comment{advance bucket~$c$ past new element}
\EndFor
\EndProcedure
\end{algorithmic}
\end{algorithm}

\goodbreak

With the input array~$A$ being consumed, output is produced in a separate array~$A'$.
While only $n$~elements
are occupied in both arrays together at any point in the algorithm, nevertheless
$2n$~elements of storage, i.\,e.~$n$~elements of overhead in addition to the input,
need to be provided.
We would like to address this shortcoming with a novel LSD radix sort algorithm
that does the job with only $\mathcal O(\sqrt n)$ extra space at a similar performance.
We call this algorithm \emph{Radsort} due to it being
a \textbf{rad}ical overhead \textbf{rad}ix sort.

\begin{algorithm}
\caption{the Radsort algorithm} \label{alg:fullsort}
\begin{algorithmic}[1]
\Procedure{radSort}{$A,n,n_t$}
\State$\Call{setup}{A,n}$\Comment{Alg.~\ref{alg:setup}}
\For{$t\gets n_t-1\ldots0$}
 \State$\Call{sortPhase}{t}$\Comment{Alg.~\ref{alg:sortphase}}
 \State$\Call{fixupPhase}{}$\Comment{Alg.~\ref{alg:fixupphase}}
\EndFor
\State$\Call{Finalize}{}$\Comment{Alg.~\ref{alg:finalize}}
\EndProcedure
\end{algorithmic}
\end{algorithm}

\section{Radsort} \label{sec:radsort}
Radsort (Alg.~\ref{alg:fullsort})
treats the input array~$A$ as a sequence of \emph{blocks} of some block size~$b$%
\footnote{see \S~\ref{sec:analysis} for discussion on the selection of~$b$}.
After a block of input is consumed, its storage is reused for subsequently
produced output, reducing the storage overhead to a constant number of
\emph{scratch blocks}, as well as some bookkeeping.

Each round of sorting comprises two \emph{phases}.
During the \emph{sort phase}, one output block is initially assigned to each of the
$\sigma$~\emph{buckets} we sort into.  Once such a block is full, a new block
is drawn from previously consumed input blocks, overwriting the input as it is processed.
The result of the sort phase are $\sigma$~randomly interleaved sequences of output blocks,
each forming one of the output buckets.
In the \emph{fixup phase} following each \emph{sort phase}, these sequences are
deinterleaved, giving the data block in order of buckets.
The deinterleaving does not move data blocks around---instead, a permutation~$\pi$
tracks the order in which the data blocks need to be processed for the data to appear
in order.

Like in a conventional LSD radix sort, these two phases are repeated for each key
position in turn.
After all~$n_t$ rounds of sorting, a~\emph{finalisation step} shuffles the data
blocks around, giving a sorted result in~$A$.  While Radsort allocates space in a
more complicated way than a traditional LSD~radix sort, it works according to the
same principles, giving the same stability guarantees at lower memory overhead and
better cache locality.

\subsection{Data Structures} \label{sec:datastructures}
In addition to the $\lfloor n/b\rfloor$~blocks of~$A$, $2\sigma$~\emph{scratch
blocks} stored in the temporary array~$T$ are required:
$\sigma$~blocks of \emph{head start} to cover the initial
output blocks for the buckets of the current round, and $\sigma$~additional
covering the partial blocks of the previous round of sorting.
The algorithm treats the concatenation of $T$ and~$A$ as one long array of
$n_\pi$~blocks of $b$~elements each, indexed through the $\blockat$~function.
The first $2\sigma$~indices refer to blocks in~$T$, the other indices to blocks overlaying~$A$.
The $\blockat(i)$ function returns a pointer to the beginning of the block at index~$i$.
\begin{equation}
\blockat(i)=
\begin{cases}
\mbox{address of $T[ib]$}&\mbox{if $0\le i<2\sigma$}\\
\mbox{address of $A[(i-2\sigma)b]$}&\mbox{if $2\sigma\le i<n_\pi$}
\end{cases}
\end{equation}
The $n\bmod b$ elements at the end of~$A$ are copied into a scratch block at
the start of the algorithm.  They are unused until the contents of~$T$
are copied back to~$A$ at the end of the algorithm.
The array~$\pi$ describes a permutation of these $n_\pi$~blocks, such that
$\pi[i]$~holds the index of the block at \emph{logical index}~$i$.
Fig.~\ref{fig:datalayout} shows an example of this data structure for strings of
$n_t=4$~characters on initialisation.

Each block is either \emph{allocated} or \emph{unallocated}.
An allocated block is either
\emph{full}, meaning that all of its elements hold data, or \emph{partial}, meaning that
some of its elements hold data.
Unallocated blocks do not hold any data and are ready for reuse.
At the first~$\sigma$ logical indices, a \emph{head start} of $\sigma$~unallocated blocks
is found.
They are followed by allocated blocks at logical indices~$\sigma\le i<f$,
and finally unallocated blocks are logical indices~$f\le i<n_\pi$.

Array~$P$ tracks up to $\sigma$~partial blocks.
It holds in order of logical index partial blocks~$(i,l)$
described by logical index~$i$ and block length $0\le l<b$.
Of each such partial block, the
first $l$~elements are used, and the other $b-l$~elements are unused.
To simplify the implementation, the block at logical index $f-1$ is always a partial block.

\begin{algorithm}
\caption{finding the size of the next block} \label{alg:blocksize}
\begin{algorithmic}[1]
\Function{sizeOfNextBlock}{$i,i_P$}
\State$(i_{i_P},l_{i_P})\gets P[i_P]$
\If{$\pi[i]=i_{i_P}$}\Comment{is $\pi[i]$ a partial block?}
 \State$i_P\gets i_P+1$
 \State\textbf{return} $l_{i_P}$
\Else\ \textbf{return} $b$
\EndIf
\EndFunction
\end{algorithmic}
\end{algorithm}

Concatenated in the order given by~$\pi$, the contents of full and partial
blocks hold the elements to be sorted in the order induced by the
rounds of sorting performed so far.
We can iterate over the allocated blocks in logical order, discovering the
length of each block by traversing~$P$ at the same time.  This idea is wrapped in the
function $\Call{sizeOfNextBlock}{i,i_P}$ which gives the length of the block at logical index~$i$,
assuming the next partial block is tracked in~$P[i_P]$.  If the block is partial, $i_P$~is
updated to track the next partial block.

In summary, the state of the Radsort algorithm is tracked in the
variables given below.
Additional variables are needed within individual algorithm
steps and are explained there.

\vspace{\abovedisplayskip}
\begin{centering}
\begin{tabular*}{\textwidth}{rcl}
\toprule
$A$&input array&array of strings\\
$n$&input length&$n=|A|$\\
$\sigma$&alphabet size&$\sigma=|\Sigma|$\\
$n_t$&key length&integer\\
$b$&block size&integer\\
$n_\pi$&block count&$n_\pi=\lfloor n/b\rfloor +2\sigma$\\
$T$&scratch buffer&array of $2\sigma b$~strings\\
$f$&fill level&integer\\
$\pi$&permutation&array of $n_\pi$ integers\\
$P$&partial blocks&array of $\sigma$ pairs of integers\\
\bottomrule
\end{tabular*}
\end{centering}

\begin{figure}[p]
\begin{centering}
\begin{tikzpicture}



  \node[draw,minimum width=3cm,anchor=west] at (0,2) (T) {\phantom{A}};
  \node[node distance=-.4pt, left=of T] {$T$};
  \path[draw,decorate,decoration={brace,raise=2pt,amplitude=5pt}]
    (T.north west) -- (T.north east) node[midway,above=5pt]{$2\sigma$};
  \node at (1.75,2) {$\cdots$};
  
  \node[draw,minimum width=7.2cm,anchor=west,node distance=-.4pt,right=of T.east,pattern=north east lines, pattern color=black] (A) {\phantom{A}};
  \node[node distance=-.4pt, right=of A] (Alabel){$A$};
  \path[draw,decorate,decoration={brace,raise=2pt,amplitude=5pt}]
    (A.north west) -- ([xshift=-2mm]A.north east) node[midway,above=5pt]  (Anb) {$\lfloor n/b\rfloor$};
  \node[fill=white,inner sep=1pt] at (6.5,2) {$\cdots$};

  \node[text width=2cm,align=center,font=\small,xshift=3mm,yshift=1mm] at (Alabel |- Anb) {$n\bmod b$\\ tail elements};
  \draw ([xshift=-1mm,yshift=2pt]A.north east) to[out=60,in=270] +(.3,.15);

  \node[draw,minimum width=6.6cm,anchor=west,node distance=-.4pt] at (1,0) (pi) {\phantom{A}};
  \node[node distance=-.4pt, left=of pi] {$\pi$};
  \path[draw,decorate,decoration={brace,raise=11pt,amplitude=5pt,mirror}]
  (pi.south west) -- +(1.6,0) node[midway,below=15pt,text width=2cm,align=center] {$\vphantom\le\sigma$\\ head start};
  \path[draw,decorate,decoration={brace,raise=11pt,amplitude=5pt}]
  (pi.south east) -- +(-2,0) node[midway,below=15pt,text width=2cm,align=center] (f) {$\le\sigma$\\ unallocated};

  \node[draw,fill=black,minimum width=3.0cm,anchor=west,node distance=-.4pt,pattern=north east lines, pattern color=black] at (2.6,0) (data) {\phantom{A}};
  \path[draw,decorate,decoration={brace,raise=11pt,amplitude=5pt}]
  (data.south east) -- (data.south west) node[midway,below=15pt,text width=2cm,align=center] (f) {$\vphantom\le$\\ data entries};

 \node[below=-2pt,font=\small,text height=1.5ex,text depth=.25ex] at (1.1,0 |- pi.south) (lindex) {0};
 \node[below=-2pt,font=\small,text height=1.5ex,text depth=.25ex] at (1.3,0 |- pi.south) {1};
 \node[below=-2pt,font=\small,text height=1.5ex,text depth=.25ex] at (2.7,0 |- pi.south) {$\sigma$};
 \node[below=-2pt,font=\small,text height=1.5ex,text depth=.25ex] at (5.7,0 |- pi.south) {$f$};
 \node[below=-2pt,font=\small,text height=1.5ex,text depth=.25ex] at (7.8,0 |- pi.south) {$n_\pi-1$};
  \node[node distance=-3.5pt, left=of lindex,font=\small,text height=1.5ex,text depth=.25ex] {logical index};
  
  \node[draw,minimum width=37mm,anchor=west,node distance=-.4pt] at (1.7,3.5) (sampleblock) {\phantom{A}};
  \node[node distance=-.4pt, left=of sampleblock] {a block with $b$ elements};
  \node[node distance=-.4pt, right=of sampleblock] {of $n_t=4$ characters each};
  \node[anchor=west,text depth=.25ex,text height=1.5ex] at (1.70,3.5) {\texttt{frob}\texttt{~\;quux}};
  \node[text depth=.25ex,text height=1.5ex] at (4.1,3.4) {$\cdots$};
  \node[anchor=west,text depth=.25ex,text height=1.5ex] at (4.35,3.5) {\texttt{buzz}};
\foreach \i in {2.75,3.72,4.4}
\draw (\i,0 |- sampleblock.south) -- (\i,1 |- sampleblock.north);
\draw[dashed,black] (sampleblock.south west) to[out=300,in=90] (3.5,2 |- A.north);
\draw[dashed,black] (sampleblock.south east) to[out=240,in=90] (4,2 |- A.north);

\foreach \i in {0.5,1.0,2.5,3.0,...,4.0,9.5,10}
\draw (\i,0 |- T.south) -- (\i,1 |- T.north);

\foreach \i in {1.2,1.4,2.4,2.6,...,3.2,5.4,5.6,...,6.2,7.4}
   \draw (\i,0 |- pi.south) -- (\i,1 |- pi.north);

\draw[thick,>=stealth,->] (1.1,0) to[out=90,in=270] (.3,2 |- A.south);
\draw[thick,>=stealth,->] (1.3,0) to[out=90,in=270] (.75,2 |- A.south);
\node at (1.9,0) {$\cdots$};
\draw[thick,>=stealth,->] (2.5,0) to[out=90,in=270] (1.5,2 |- A.south);

\draw[thick,>=stealth,->] (2.7,0) to[out=90,in=270] (3.3,2 |- A.south);
\draw[thick,>=stealth,->] (2.9,0) to[out=90,in=270] (3.75,2 |- A.south);
\node[fill=white,inner sep=1pt] at (4.1,0) {$\cdots$};
\draw[thick,>=stealth,->,looseness=.6] (5.5,0) to[out=90,in=270] (9.75,2 |- A.south);

\draw[thick,>=stealth,->,looseness=.4] (5.7,0) to[out=90,in=270] (1.8,2 |- A.south);
\draw[thick,>=stealth,->,looseness=.4] (5.9,0) to[out=90,in=270] (2.1,2 |- A.south);
\node at (6.6,0) {$\cdots$};
\draw[thick,>=stealth,->,looseness=.4] (7.5,0) to[out=90,in=270] (2.8,2 |- A.south);

\end{tikzpicture}
\end{centering}
\caption{Radsort's data structures at initialisation} \label{fig:datalayout}
\end{figure}

\begin{figure}[p]
\centering%
\input fig-sortphase
\caption{Algorithm state during the second round of sorting an array of strings
with $n=26$, $n_t=4$, $b=4$, and $\Sigma=\{\texttt a, \texttt b, \texttt c\}$.
Changes highlighted in red.
Unallocated cells are blanked out.} \label{fig:sortphase}
\end{figure}

\subsection{Initialisation} \label{sec:init}
The variables are initially set up (cf.~Fig.~\ref{fig:datalayout}) such that the contents of~$A$ are represented by
the blocks at logical indices $\sigma$ to~$\sigma+\lceil n/b\rceil$.
The final $n\bmod b$~elements of~$A$ are copied to block~$\sigma$.
This establishes the \emph{data invariant} of the state variables,
shown below.  Each round of sorting assumes that this invariant holds
initially and reestablishes it at conclusion.

\begin{enumerate}
\item $\pi$ is a permutation of $0\ldots n_\pi-1$
\item $P$ tracks the partial blocks in ascending order of logical index
\item the block at logical index $f-1$ is a partial block
\item logical indices $0$ to~$\sigma-1$ and $f$ to$~n_\pi-1$ refer to unallocated blocks
\item logical indices $\sigma$ to~$f-1$ refer to allocated blocks, the concatenation of whose
used elements holds the input permuted according to the current progress of the algorithm
\end{enumerate}

\begin{algorithm}
\caption{initialisation of variables}\label{alg:setup}
\begin{algorithmic}[1]
\Procedure{setup}{$A,n$}
\State$f\gets\sigma+\lfloor n/b\rfloor+1$
\State$\pi[0\ldots\sigma-1]\gets0,\ldots,\sigma-1$\Comment{assign head start}
\State$\pi[\sigma\ldots f-2]\gets2\sigma,\ldots,\sigma+f-2$\Comment{assign contents of~$A$}
\State$\pi[f-1\ldots n_\pi-1]\gets\sigma,\ldots,2\sigma-1$\Comment{assign remaining blocks}
\State$T[0\dots2\sigma b-1]\gets0$\Comment{filly~$T$ with dummy values (optional)}
\State$\blockat(\sigma)[0\dots(n\bmod b)-1]\gets A[b\lfloor n/b\rfloor\ldots n-1]$\Comment{copy tail to scratch block}
\State$P[0]\gets(f-1,n\bmod b)$\Comment{track input tail as partial block}
\State$P[1\ldots\sigma-1]\gets(n_\pi,0)$\Comment{fill $P$ with dummy values (optional)}
\EndProcedure
\end{algorithmic}
\end{algorithm}

\subsection{Sorting}
Each round of sorting (cf.~Fig.~\ref{fig:sortphase}) sorts the elements of~$A$
stably according to some key position~$t$.
We write $\key(A[i],t)\in\Sigma$ for key position~$t$ of element~$A[i]$.
First, the \emph{sort phase} sorts elements into
buckets according to the value of their key.
Then, the \emph{fixup phase} computes new $\pi$ and~$P$ based on the results of the sort phase,
restoring the data invariant.
The sort and fixup phases require several additional variables that can be discarded afterwards:

\vspace{\abovedisplayskip}
\begin{centering}
\begin{tabular}{rcl}
\toprule
$B$&bucket allocation&array of $\sigma$~pairs of pointers\\
$C$&block counts&array of $\sigma$ integers\\
$S$&bucket starts&array of $\sigma$ integers\\
$U$&block usage&array of $n_\pi$ characters\\
$\pi'$&new permutation&array of $n_\pi$ integers\\
$i_{\mathrm{in}}$&next input block&integer\\
$i_{\mathrm{out}}$&next output block&integer\\
$i_P$&next partial block&integer\\
\bottomrule
\end{tabular}
\end{centering}

\subsubsection{Sort phase}
The sort phase allocates one block for each bucket in the bucket array~$B$.  This array
holds for each bucket~$c$ a pair of pointers $(p_{\mathrm{next}},p_{\mathrm{end}})$
pointing to the next free element, and the end of the block.
It then traverses the allocated blocks by~$i_{\mathrm{in}}$ in logical order, and for
each block sorts its elements into the correct output block.  If an output block is full,
a new block is drawn from the next logical index~$i_{\mathrm{out}}$ using
the~$\Call{newBucketBlock}{c}$ procedure.
The bucket each output block is
used for is tracked in the usage array~$U$, and the number of blocks assigned to each
of the buckets is tracked in~$C$.

As the sort proceeds through the input array, consumed blocks of input are reassigned into
output blocks.  It is critical that blocks of input are fully consumed before they are reassigned
into output blocks, lest elements are overwritten before they can be sorted.  This is ensured
through $i_{\mathrm{in}}$~having a \emph{head start} of $\sigma$~unallocated blocks:

\begin{lemma}
When a new output block needs to be assigned, $i_{\mathrm{out}}<i_{\mathrm{in}}$.
\end{lemma}
\begin{proof}
Let $j$~be the number of elements processed so far.
The input comprises of blocks holding at most $b$~elements, so at least
$\lfloor j/b\rfloor$ blocks of input have been processed so far, hence
$\sigma+\lfloor j/b\rfloor\le i_{\mathrm{in}}$.
As for output blocks, a new one is assigned whenever the current block of some bucket
is full.  Thus there are now $\sigma-1$ partial output blocks (one for
each bucket, except for the bucket we need to assign a new output block to)
as well as up to $\lfloor j/b\rfloor$ full output blocks.  Hence
$i_{\mathrm{out}}\le\sigma-1+\lfloor j/b\rfloor$.  Joining these two
together gives
\begin{equation}
i_{\mathrm{out}}\le\sigma-1+\lfloor j/b\rfloor<\sigma+\lfloor j/b\rfloor\le i_{\mathrm{in}}
\end{equation}
as required.
\end{proof}

\begin{algorithm}
\caption{sort phase}\label{alg:sortphase}
\begin{algorithmic}[1]
\Procedure{newBucketBlock}{$c$}\Comment{allocate a new bucket block for key~$c$}
 \State$p_{\mathrm{next}}\gets\blockat(\pi[i_{\mathrm{out}}])$
 \State$B[c]\gets(p_{\mathrm{next}},p_{\mathrm{next}}+b)$
 \State$C[c]\gets C[c]+1$
 \State$U[i_{\mathrm{out}}]\gets c$
 \State$i_{\mathrm{out}}\gets i_{\mathrm{out}}+1$
\EndProcedure
\Procedure{sortPhase}{t}
\State$C[0\ldots\sigma-1]\gets i_P\gets i_{\mathrm{out}}\gets0$\Comment{initialise variables}
\State\textbf{for $c\in\Sigma$ do} \Call{newBucketBlock}{c}\Comment{assign initial output blocks}
\For{$i_{\mathrm{in}}\gets\sigma\ldots f-1$}
 \State$p_{\mathrm{in}}\gets\blockat(\pi[i_{\mathrm{in}}])$\Comment{get next block of input}
 \State$l\gets\Call{sizeOfNextBlock}{i_{\mathrm{in}},i_P}$
 \For{$j\gets 0\ldots l-1$}
  \State$c\gets\key(p_{\mathrm{in}}[j],t)$\Comment{find bucket of current element}
  \State$(p_{\mathrm{next}},p_{\mathrm{end}})\gets B[c]$
  \State copy $p_{\mathrm{in}}[j]$ to memory at $p_{\mathrm{next}}$ \label{stmt:elemwrite}
  \If{$p_{\mathrm{next}}+1=p_{\mathrm{end}}$}\Comment{bucket full?}
   \State$\Call{newBucketBlock}{c}$
  \Else
   \State$B[c]\gets(p_{\mathrm{next}}+1,p_{\mathrm{end}})$
  \EndIf
 \EndFor
\EndFor
\EndProcedure
\end{algorithmic}
\end{algorithm}

\subsubsection{Fixup phase}
After the sort phase, the blocks at logical index $0$ to~$i_{\mathrm{out}}$ hold the
sorted buckets in some interleaving, with up to $b$~partial blocks described by~$B$.
The fixup phase computes a new permutation~$\pi'$ that starts with
a head start of $\sigma$~unallocated blocks, followed by the deinterleaved buckets in order
of keys and finally any remaining unallocated blocks.
We are guaranteed to always find $\sigma$~unallocated blocks for the head start:

\begin{lemma}
After the sort phase, there are at least $\sigma$~unallocated blocks.
\end{lemma}
\begin{proof}
A full block holds $b$~elements.  As there are $n$~elements in total, there are at most
$\lfloor n/b\rfloor$~full output blocks.  Furthermore, there are $\sigma$~partial blocks
tracked in~$B$.  As each block after sorting is either used as an output block or unallocated,
and as there are $n_\pi$~blocks in total, we find for the number of unallocated blocks~$n_{\mathrm{unallocated}}$
\begin{equation}
n_{\mathrm{unallocated}}\ge n_\pi-\sigma-\lfloor n/b\rfloor=\sigma
\end{equation}
as required.
\end{proof}

\noindent
The buckets are deinterleaved by first computing the starting
indices~$S[0\ldots\sigma-1]$ of the buckets and then by shuffling the blocks according
to the usages~$U$ into~$\pi'$.  Meanwhile, each allocated block is compared with the
pointers in~$B$ to populate~$P$ with the whereabouts of the new partial blocks.  As there is
one partial block for each bucket, each of the $\sigma$~elements of~$P$ ends up being initialised.
Finally, $\pi$ is set to $\pi'$ to restore the data invariant.

\begin{algorithm}
\caption{fixup phase}\label{alg:fixupphase}
\begin{algorithmic}[1]
\Procedure{fixupPhase}{}
\State$S[0]\gets\sigma$
\State\textbf{for $i\gets 1\ldots\sigma-1$ do} $S[i]\gets S[i-1]+C[i-1]$\Comment{prefix sum over $C$ plus $\sigma$}
\State$\pi'[0\ldots\sigma-1]\gets\pi[i_{\mathrm{out}}\ldots i_{\mathrm{out}}+\sigma-1]$\Comment{assign head start}
\For{$i\gets0\ldots i_{\mathrm{out}}-1$}\Comment{shuffle used blocks into order of buckets}
 \State$c\gets U[i]$
 \State$\pi'[S[c]]\gets\pi[i]$
 \State$(p_{\mathrm{next}},p_{\mathrm{end}})\gets B[c]$
 \If{$p_{\mathrm{end}}=\blockat(\pi[i])+b$}\Comment{is $\pi[i]$ a partial block?} \label{stmt:ispartialblock}
  \State$P[c]\gets\bigl(S[c],p_{\mathrm{next}}-\blockat(\pi[i])\bigr)$
 \EndIf
 \State$S[c]\gets S[c]+1$
\EndFor
\State$\pi'[i_{\mathrm{out}}+\sigma\ldots n_\pi-1]\gets\pi[i_{\mathrm{out}}+\sigma\ldots n_\pi-1]$\Comment{assign remaining blocks}
\State$f\gets i_{\mathrm{out}}+\sigma$
\State$\pi\gets\pi'$
\EndProcedure
\end{algorithmic}
\end{algorithm}

\begin{algorithm}[t]
\caption{restoration of the contents of $A$} \label{alg:finalize}
\begin{algorithmic}[1]
\Procedure{swap}{$i_1,j_1,i_2,j_2$}
  \State$\pi[i_1],\pi[i_2]\gets j_2,j_1$
  \State$\pi^{-1}[j_1],\pi^{-1}[j_2]\gets i_2,i_1$
  \State$i_1,i_2\gets i_2,i_1$
\EndProcedure
\Procedure{finalize}{}
\State$k\gets i_P\gets0$
\State$i_{\mathrm{free}},j_{\mathrm{free}}\gets0,\pi[0]$\Comment{guaranteed unallocated as per invariant}
\State\textbf{for $i\gets 0\dots n_\pi-1$ do $\pi^{-1}[\pi[i]]\gets i$}\Comment{compute inverse permutation}
\For{$i_{\mathrm{in}}\gets\sigma\ldots\sigma+\lfloor n/b\rfloor-1$}\Comment{pull first $\lfloor n/b\rfloor$ blocks}
 \State$j_{\mathrm{in}}\gets\pi[i_{\mathrm{in}}]$
 \State$j_{\mathrm{out}},i_{\mathrm{out}}\gets i_{\mathrm{in}}-\sigma,\pi^{-1}[i_{\mathrm{in}}-\sigma]$
 \If{$i_{\mathrm{in}}<i_{\mathrm{out}}<f$}\Comment{$j_{\mathrm{out}}$ not the right block, but allocated?}
  \State copy block at $j_{\mathrm{out}}$ to block at $j_{\mathrm{free}}$\Comment{push block at $j_{\mathrm{out}}$ away}
  \State$\Call{swap}{i_{\mathrm{out}},j_{\mathrm{out}},i_{\mathrm{free}},j_{\mathrm{free}}}$\Comment{track swapping of $i_{\mathrm{free}}$ and $i_{\mathrm{out}}$}
 \EndIf
 \State$l\gets\Call{sizeOfNextBlock}{i_{\mathrm{in}},i_P}$
 \State$A[k\ldots k+l-1]\gets\blockat(j_{\mathrm{in}})[0\ldots l-1]$\Comment{pull block from $j_{\mathrm{in}}$}
 \State$k\gets k+l$
  \State$\Call{swap}{i_{\mathrm{out}},j_{\mathrm{out}},i_{\mathrm{in}},j_{\mathrm{in}}}$\Comment{track swapping of $i_{\mathrm{in}}$ and $i_{\mathrm{out}}$}
 \If{$i_{\mathrm{in}}\ne i_{\mathrm{out}}$}\Comment{did we pull from elsewhere?}
   \State$i_{\mathrm{free}},j_{\mathrm{free}}\gets i_{\mathrm{in}},j_{\mathrm{in}}$\Comment{if yes, there is now a new free spot}
 \EndIf
\EndFor
\For{$i_{\mathrm{in}}\gets\sigma+\lfloor n/b\rfloor\ldots f-1$}\Comment{pull any remaining blocks (all in $T$)}
 \State$j_{\mathrm{in}}\gets\pi[i_{\mathrm{in}}]$
 \State$l\gets\Call{sizeOfNextBlock}{i_{\mathrm{in}},i_P}$
 \State$A[k\ldots k+l-1]\gets\blockat(j_{\mathrm{in}})[0\ldots l-1]$\Comment{pull block from $j_{\mathrm{in}}$}
 \State$k\gets k+l$
\EndFor
\EndProcedure
\end{algorithmic}
\end{algorithm}

\subsection{Finalisation}

Following some rounds of sorting, the used elements of the blocks at logical indices $\sigma$ to~$f-1$
hold the elements of the input in sorted order.  To wrap things up, we need to shuffle
this collection of logically ordered blocks back into~$A$, preserving the ordering.

This is achieved by shuffling the used blocks around such that they are found at indices~$2\sigma\ldots 2\sigma+f-1$,
When the blocks are moved into place, their elements
are actually moved right behind the elements of the previous block, ignoring gaps of
unused elements at the end of partial blocks, and leaving the elements in their
final locations.
After $\lfloor n/b\rfloor$~blocks have been shuffled to indices $2\sigma$ to~$n_\pi-1$,
all remaining blocks must be located in~$T$.
As an optimisation, we can directly copy them to their final positions without
moving other data away, as there is no risk of other blocks being in their way.

In this step, pairs of indices~$i_\circ,j_\circ$ refer to the
logical and physical indices of the same block.
These are related through the identity
\begin{equation}
j_\circ=\pi[i_\circ],\quad i_\circ=\pi^{-1}[j_\circ].
\end{equation}
The following extra variables are used:

\vspace{\abovedisplayskip}
\begin{centering}
\begin{tabular}{rcl}
\toprule
$\pi^{-1}$&inverse permutation&array of $n_\pi$ integers\\
$i_{\mathrm{in}}, j_{\mathrm{in}}$&block holding next data chunk&integers\\
$i_{\mathrm{out}}, j_{\mathrm{out}}$&where next block goes&integers\\
$i_{\mathrm{free}}, j_{\mathrm{free}}$&some unallocated block&integers\\
$i_P$&next partial block&integer\\
$k$&no.~of elements written into $A$&integer\\
\bottomrule
\end{tabular}
\vspace{\belowdisplayskip}
\end{centering}

To help us carry out these shuffles, we first compute an inverse permutation $\pi^{-1}$ such
that $\pi\circ\pi^{-1}=(0,\ldots,n_\pi-1)$.  The assignment of two blocks $j_1$ and~$j_2$ can be
swapped with the $\Call{swap}{i_1,j_1,i_2,j_2}$ function, while moving their contents in a separate
step.

For each block index~$2\sigma\le j<n_\pi$, we
distinguish three cases: (a)~the right block is already at index~$j$, in which case we
just move it ahead to fill the gap,  (b)~the block is unallocated, in which case we find the
right block at~$\pi^{-1}[j]$ and \emph{pull} it to its final position, or (c)~there is some other
allocated block at index~$j$, in which case we \emph{push} it to some free block and reduce to
case~(b).  The variable~$j_{\mathrm{free}}$ tracks a spot to push an empty block to.
whenever this spot is occupied, we are going to subsequently pull a block from
somewhere else, restoring the unallocated spot.


\section{Complexity Analysis} \label{sec:analysis}
For analysis, we assume array element size and alphabet size~$\sigma$ to be
constants and $b\in\mathcal O(n)$.

\subparagraph*{Runtime}
The $\Call{setup}{A,n}$ procedure from Alg.~\ref{alg:setup} initialises
arrays of $n_\pi\in\mathcal O(n/b)$ and~$b$ elements respectively for a
total runtime of~$\mathcal O(n/b+b)\subset\mathcal O(n)$.
The $\Call{sortPhase}{t}$ procedure from Alg.~\ref{alg:sortphase} executes
one inner loop iteration for each of the $n$~array elements for a
total runtime of~$\mathcal O(n)$.
The $\Call{fixupPhase}{}$ procedure from Alg.~\ref{alg:fixupphase}
writes to each element of~$\pi'$ once, for a total runtime of~$\mathcal O(n/b)$.
Lastly, the $\Call{finalize}{}$ procedure from Alg.~\ref{alg:finalize}
traverses $\pi$, copying $\lfloor n/b\rfloor$~blocks at most twice,
while copying the remaining $f-\lfloor n/b\rfloor-\sigma\le\sigma$ blocks
at most once for a total runtime of~$\mathcal O(b(n/b+\sigma))\subset\mathcal O(n)$.
For the full sort given in Alg.~\ref{alg:fullsort}, the $\Call{sortPhase}{t}$
and~$\Call{fixupPhase}{}$ procedures are called once for each key position,
raising the total runtime to $\mathcal O(n_t n)$, which matches the conventional
LSD radix sort.

\subparagraph*{Space}
In addition to the input stored in~$A$, array $T$ occupies
$\mathcal O(b)$~bytes of storage, arrays~$U$, $\pi$, $\pi'$, and~$\pi^{-1}$ occupy
$\mathcal O(n/b)$ bytes of storage, and the other variables are of constant size,
for a total extra storage requirement of $\mathcal O(b+n/b)$.  This requirement
can be minimised by choosing some $b\in\Theta(\sqrt n)$,
giving just $\mathcal O(\sqrt n)$~bytes of extra storage.

\section{Implementation Notes and Optimisations} \label{sec:implnotes}
We have produced a variety of Radsort implementations in the
C~language.  We make our code freely available.%
\footnote{See \url{https://github.com/clausecker/radsort}.}
Implementations in other languages should be straightforward; in languages
that do not provide pointers or pointer arithmetic, the $B$~array can be
altered to use the same data structure as the $P$~array, although at a loss
of performance.

Using the pointer-based data structure of the $B$~array for
the $P$~array is not advisable: While $\Call{fixupPhase}{}$ (Alg.~\ref{alg:fixupphase})
is simplified by not needing to translate $B$ into~$P$, performance
of~$\Call{sortPhase}{t}$ (Alg.~\ref{alg:sortphase}) is reduced due
to the longer dependency chain when checking if the next input block
is partial.  Additionally, the $\Call{finalize}{}$~procedure
(Alg.~\ref{alg:finalize}) needs partial blocks to be tracked by their
logical indices, mandating that the data structure be translated at
least once.

A radix of~$\sigma=256$ appears to be optimal for sorting key-value
pairs of integers, though the best choice depends on cache and
element size, as well as key type and distribution and must be
determined empirically.  The basic tradeoff is that larger radices
allow for less rounds of sorting, but require more cache to track the
buckets, eventually escalating into higher levels of the cache hierarchy
at great performance cost.

For a type-generic implementation, the $\Call{sortPhase}{t}$~procedure
should be monomorphised, ideally for each key position.  In the common
case of integer keys split into byte-sized digits ($\sigma=256$),
retrieving the byte at position~$t$ of the key by means of shifts and
masks performs better than direct loading of the key byte from memory,
as the latter generally incurs an extra memory operation, while the
former shifts the load to arithmetic execution units that are otherwise
underutilised.

While a block size of order $b\in\Theta(\sqrt n)$ minimises space overhead,
choosing a fixed block size simplifies the implementation.
The space overhead then becomes some
fixed fraction of the input size, which is often acceptable.  For example,
we used~$b=512$ in our experiments.  At an element size of~$8\,\textrm B$,
this gives a fixed overhead of~$2\,\textrm{MiB}$ for the $T$~array,
$8\,\textrm{KiB}$~for arrays~$B$,
$C$, $P$, and~$S$ together, and $9\,\textrm B$ per block of input
(i.\,e.~0.22\,\%~the input data set) to store $\pi$, $\pi'$, $\pi^{-1}$,
and~$U$, reusing $\pi'$ for~$\pi^{-1}$.

Like with a conventional LSD radix sort, performance of Radsort is
memory bound and eliminating just a single load or store from the
inner loop of $\Call{SortPhase}{t}$ can improve performance dramatically,
as those memory accesses compete for resources with the write
in Alg.~\ref{alg:sortphase}, L.~\ref{stmt:elemwrite}, the bottleneck of the
algorithm.  One technique to eliminate such a load is discussed
in~\S~\ref{sec:fasteob}, and gives
a 5--50\,\% speedup depending on array size and microarchitecture.

\subsection{Avoiding Finalisation}
In some use cases, it suffices to iterate over the sorted data.  In such
cases, the $\Call{finalize}{}$~procedure can be skipped, employing instead
the $\Call{sizeOfNextBlock}{i,i_P}$~function from Alg.~\ref{alg:blocksize} to
iterate over the array elements in permuted representation as in
Alg.~\ref{alg:sortphase} (see also \S~\ref{sec:datastructures}).

Using such an iterator pattern, the dataset can be modified in place,
and even be truncated (cf.~Alg.~\ref{alg:truncate}) without
disturbing the Radsort algorithm state.  This enables algorithm designs
where the same dataset is repeatedly sorted by different keys, and then
modified, without having to finalise the sort or recreate the algorithm state
for every sort.

\begin{algorithm}
\caption{truncate the dataset such that logical block~$i$, element~$j$ is the final datum} \label{alg:truncate}
\begin{algorithmic}
\Procedure{truncate}{$i,j$}
\State$f\gets i+1$ \Comment{mark logical block~$i$ as the final block}
\For{$i_P\gets0\ldots\sigma-1$}\Comment{skip partial blocks before logical block~$i$}
 \State$(i_{i_P},l_{i_P})\gets P[i_P]$
 \If{$i\le i_{i_P}$}
  \State$P[i_P]\gets(i,j+1)$\Comment{logical block~$i$ ends at element~$j$}
  \State\textbf{return}
 \EndIf
\EndFor
\EndProcedure
\end{algorithmic}
\end{algorithm}

\subsection{Faster End-of-Block Checking} \label{sec:fasteob}
The $\Call{sortPhase}{t}$~procedure of Alg.~\ref{alg:sortphase} tracks
output blocks through pairs~$(p_{\mathrm{next}},p_{\mathrm{end}})$
pointing to the next free element and the end of the block, respectively.
It is beneficial to reduce this pair to just
one pointer, as to eliminate a load of the second pointer
in the hot loop, and to reduce cache pressure.

In a simplified model, where each element is one byte
in size and the allocation of all data structures including~$A$ can be
controlled, this can be achieved by choosing some block size~$b=2^q$
and ensuring both $A$ and~$T$ are aligned to a multiple of~$b$.
Then, no~$p_{\mathrm{end}}$ is needed, as a pointer into some block
of $A$ or~$T$ points to a block boundary iff it is aligned to a
multiple of $b$~bytes---which can easily be checked with a
bitwise-and on the pointer.\footnote{The $p_{\mathrm{end}}=\blockat(\pi[i])+b$
check of Alg.~\ref{alg:fixupphase}, L.~\ref{stmt:ispartialblock} must then be
realised with a range check like $\blockat(\pi[i])\le p_{\mathrm{next}}<\blockat(\pi[i])+b$.
This is not strictly legal in the C23 language~\cite[\S~6.5.8 \P~6]{c23},
(pointers may only be compared for ordering if they point into the same array,
but $p_{\mathrm{next}}$ may point into either $A$ or~$T$)
but unproblematic in practice.}

In real applications, we usually cannot control the alignment of~$A$
and neither is the element size guaranteed to be some nice number.
However, with some modifications, the same idea
still works: let the element size be~$2^r e$, with $e$~odd and the
block size again be~$b=2^q$.  Then in an array of such elements of
no particular alignment, every $b$~elements there is an element with
address~$p_0$ such that
\begin{equation}
p_0\gg r\equiv 0\mod b, \label{eq:aligned}
\end{equation}
where $\gg$ denotes a bitwise logical right-shift.
This condition is just as easy to check using a bitwise-and on the
pointer and allows us to find block boundaries if the first element
of each block satisfies Eq.~\ref{eq:aligned}.
Furthermore, for the address of an arbitrary array element~$p$, we can find
the offset~$0\le o<b$ of such an element from~$p$ as
\begin{equation}
o=-(p\gg r)e^{-1}\bmod b \label{eq:alignoffset}
\end{equation}
where $e^{-1}$~is the modular inverse of~$e$ modulo~$b$.
Due to $b$~being a power of two, $e^{-1}$~is easy to compute~\cite{Hurchalla2022}.
This gives the following modifications: allocate
$2\sigma+1$~blocks\footnote{An extra scratch block is needed as the total number
of elements shunted to head and tail may be up to~$2b-2$, possibly exceeding a
whole block of elements.} for~$T$
aligned to a multiple of~$2^{qr}$, ensuring that its block address satisfy Eq.~\ref{eq:aligned}.
Compute the offset~$o$ of the first element of~$A$ satisfying Eq.~\ref{eq:aligned}
using Eq.~\ref{eq:alignoffset}.
The $\blockat(i)$~function is adjusted to
\begin{equation}
\blockat(i)=
\begin{cases}
\mbox{address of $T[i]$}&\mbox{if $0\le i<2\sigma+1$}\\
\mbox{address of $A[(i-2\sigma-1)b+o]$}&\mbox{if $2\sigma+1\le i<n_\pi$}
\end{cases}
\end{equation}
and~$n_\pi=\lfloor(n-o)/b\rfloor+2\sigma+1$.
Transfer the first $o$~elements to a scratch block and adjust Alg.~\ref{alg:setup}
such that this scratch block (the head) ends up at logical block~$\sigma$,
followed by the blocks overlaying~$A$, following the tail.
Various other steps of the algorithm must also be adjusted to account for the
extra scratch block.
The end-of-block check then becomes a check for the condition of Eq.~\ref{eq:aligned},
which is realised by masking the pointer with a bitmask of the form~$(b-1)\ll r$
and checking if the result is zero.  The $p_{\mathrm{end}}$~pointers no longer
need to be read and can be eliminated entirely.

This approach yielded a 5--50\,\% speedup depending on array size
and microarchitecture.
It should be considered if the programming language and environment permit it.

\subsection{Parallel Operation}
The algorithm can be parallelised with some changes to the data structures.
Before each sort phase, we distribute the blocks of the input array into
$n_t$~roughly evenly sized chunks, for sorting with $n_t$~threads.
Each chunk gets its own head start of $\sigma$~blocks,
requiring a $T$~array of $2\sigma n_t$~scratch blocks in total.
The sort phase then proceeds in parallel with each thread sorting its chunk
with its own $B$~array.  The fixup phase is sequential and must interleave the
individual thread's $B$~arrays into one global $P$~array of $n_t\sigma$~entries each.

The finalisation procedure is more troublesome to parallelise.  It can be
left as a sequential operation, as it only contributes a small amount of the total
run time, or it can be implemented by first shuffling all the blocks into order using
a parallel permutation algorithm~\cite{Hagerup1995}, and then moving the array
elements to eliminate gaps caused by partial blocks.

As Radsort is memory bound and exhibits only moderate cache locality,
a performance ceiling quickly is reached as the available memory channels
are saturated.  See \S~\ref{sec:results} for more discussion.
It may be of interest to use an initial round of sorting on the most
significant key position to split the input into buckets that can be parcelled
into threads, each of which processes the remaining key positions of its
bucket(s) from the least-significant digit sequentially.  This reduces the
working set of each thread from~$n$ to $n/n_t$~elements on
average, improving cache efficacy and reducing NUMA effects.

\section{Evaluation}
We evaluated the performance of Radsort in comparison to a classic out-of-place
LSD radix sort on a variety of systems.
Five algorithm variants are compared in total:

\begin{description}
\item[generic]
A generic out-of-place LSD radix sort (Alg.~\ref{alg:ooplsd}) with
output prefetching~\cite{Downs2019}.  An initial pass takes histograms
of the 4~key bytes, followed by 4~rounds of sorting.
\item[swc]
An out-of-place LSD radix sort implemented with Wassenberg's
soft\-ware-de\-fi\-ned write-com\-bin\-ing~\cite{Wassenberg2010}
with a block size of~$512\,\mathrm B$.
\item[radsort]
Single-threaded Radsort (Alg.~\ref{alg:fullsort}) as described in \S~\ref{sec:radsort} with
none of the improvements mentioned in \S~\ref{sec:implnotes} using a fixed
block size of~$b=512$ (i.\,e.~$4\,\mathrm{KiB}$).
\item[bitmanip]
Single-threaded Radsort implemented using the bit-manipulation based end-of-block
checking described in \S~\ref{sec:fasteob}, otherwise the same as \emph{radsort}.
\item[parallel]
Multi-threaded Radsort using 4~threads by default.  Otherwise
the same as~\emph{radsort}.
\end{description}

\subsection{Setup}
All benchmarks sort pairs of 32-bit keys and 32-bit values in $n_t=4$~rounds
by each key byte in turn~($\sigma=256$).
The key values are uniformly distributed over the
range~$0\ldots2^{32}-1$, as generated by a xorshift RNG~\cite{Marsaglia2003}.
While we have evaluated radsort on a variety of platforms, we have
selected benchmark results as measured on the following machines for this
paper:

\vspace{\abovedisplayskip}
\def\tblheading#1{\textcolor{lipicsGray}{\sffamily\bfseries\mathversion{bold}#1}}
\begin{centering}
\begin{tabular}{rccc}
\toprule
&\tblheading{power}&\tblheading{icelake}&\tblheading{grace}\\
\midrule
\emph{architecture}&powerpc64le&amd64&aarch64\\
\emph{CPU}&IBM POWER9&Intel Xeon Gold 6338&ARM Neoverse V2\\
\emph{sockets/cores/threads}&2 / 16 / 4&2 / 32 / 2&2 / 72 / 1\\
\emph{threads total}&128&128&144\\
\emph{clock speed}&$2.9\,\mathrm{GHz}$&$2.0\,\mathrm{GHz}$&$3.4\,\mathrm{GHz}$\\
\emph{L1D/L2 cache per core}&$32\,\mathrm{KiB}$ / $512\,\mathrm{KiB}$&$48\,\mathrm{KiB}$ / $1280\,\mathrm{KiB}$&$64\,\mathrm{KiB}$ / $1024\,\mathrm{KiB}$\\
\emph{L3 cache}&$10\,\mathrm{MiB}$ per 2 cores&$48\,\mathrm{MiB}$ per socket&$114\,\mathrm{MiB}$ per socket\\
\bottomrule
\end{tabular}
\end{centering}
\vspace{\belowdisplayskip}

Two sets of benchmarks were performed: in the \emph{size}~sets
(Fig.~\ref{fig:sizebench}),
all algorithms were measured on arrays with~$2^n$ and $3\times2^{n-1}$ elements
with array sizes from $2$ to as many elements as the memory fits.
In the \emph{threads}~sets (Fig.~\ref{fig:threadbench}), the \emph{parallel}
implementation was measured with an array size of~$128\,\mathrm{GiB}$
(i.\,e.~$2^{34}$~elements) and thread counts from 1 to~128.

20~runs of each benchmark were performed.
In the plots, dots show the individual runs, with a line drawn through the average.
For the size sets, execution was pinned to a single socket to avoid NUMA effects
as much as possible.

\begin{figure}[ht]
  \includegraphics[width=\textwidth]{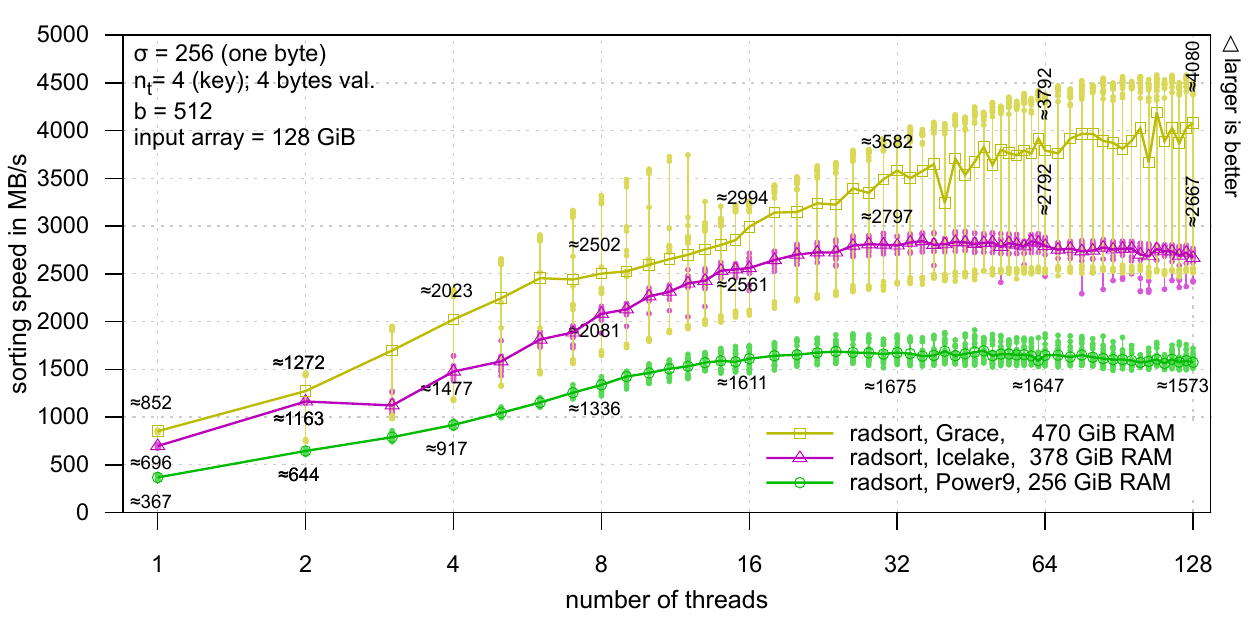}
  \caption{Sorting speed compared on a $128\,\mathrm{GiB}$~array
    with thread count from 1 to~128.}
  \label{fig:threadbench}
\end{figure}

\begin{figure}[htp]
\begin{centering}
  \includegraphics[width=\textwidth]{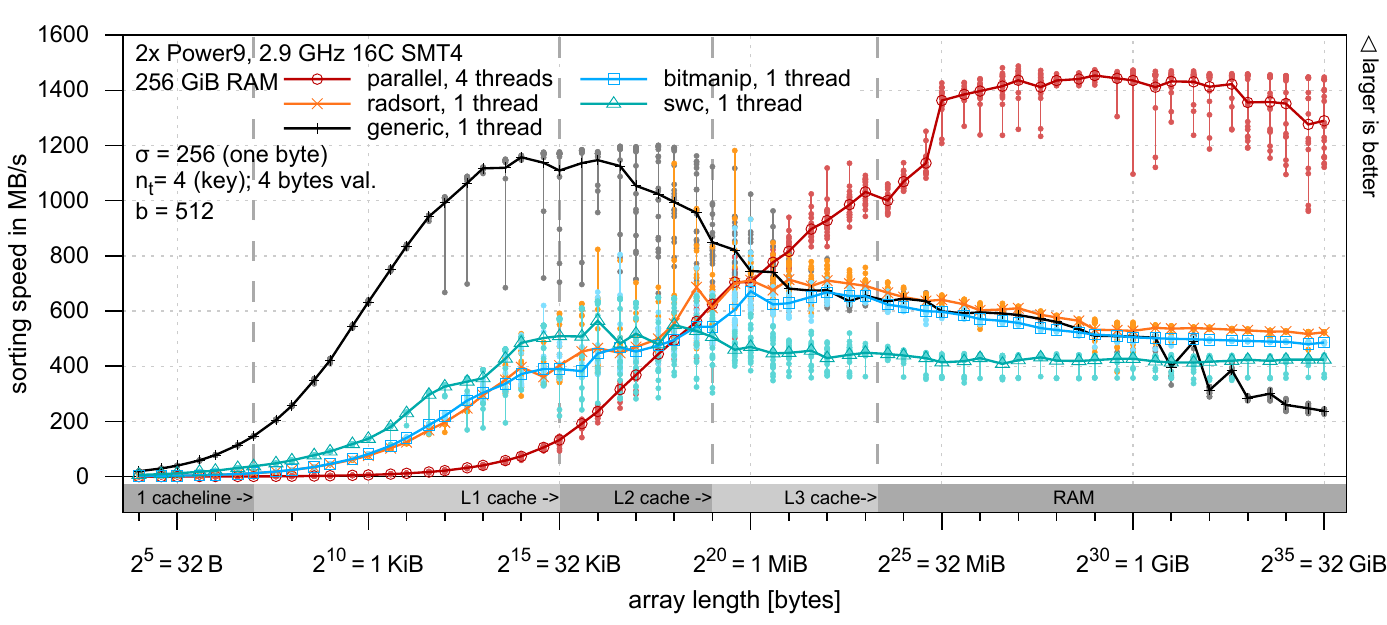}\\
  \includegraphics[width=\textwidth]{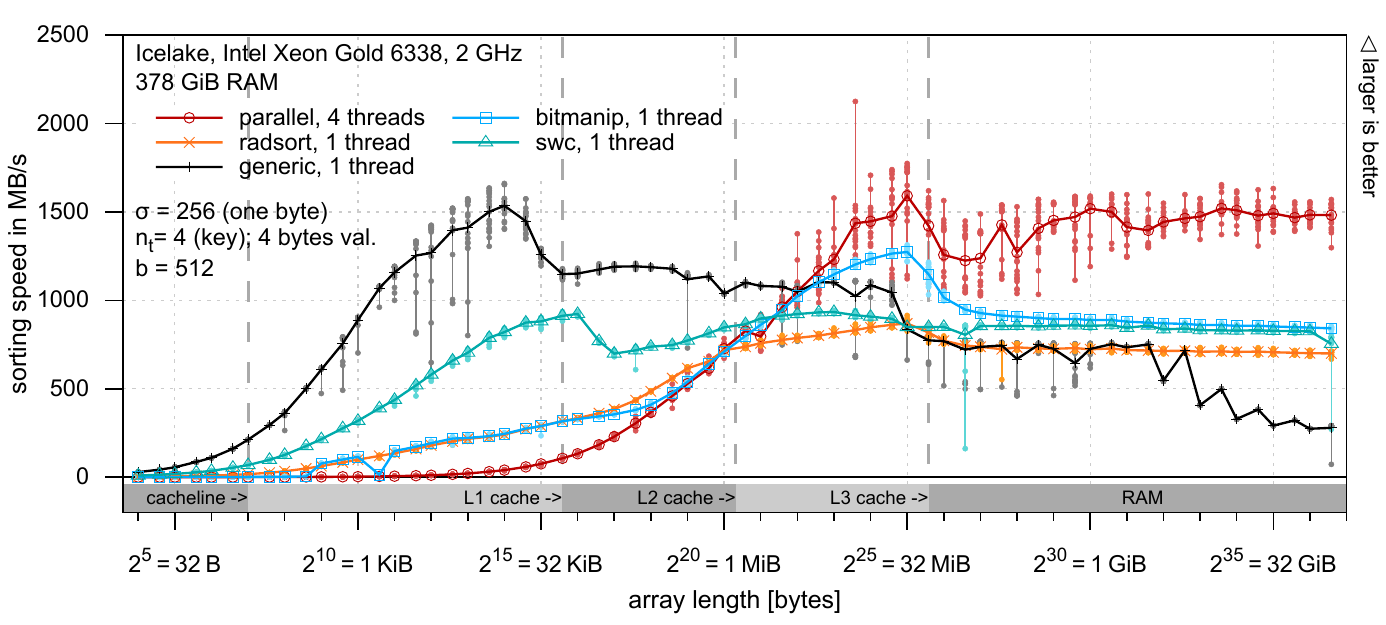}\\
  \includegraphics[width=\textwidth]{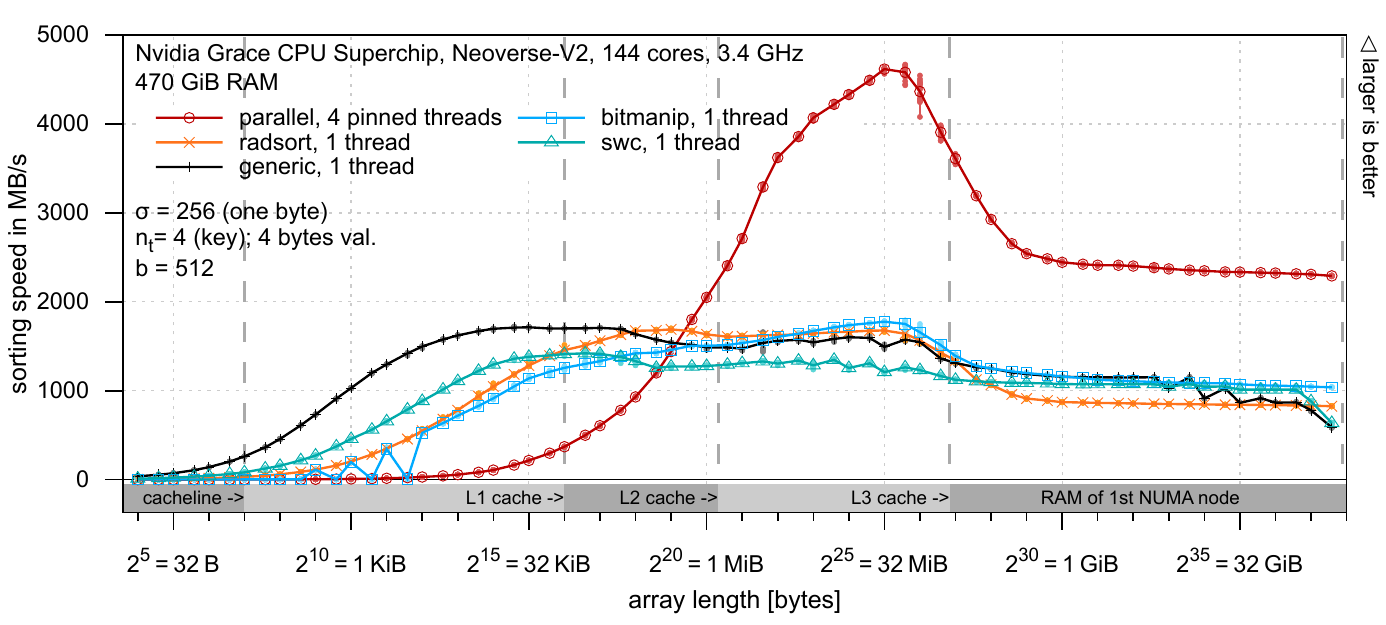}\par
\end{centering}
  \caption{Sorting speeds on three machines with arrays of
  $8\,\mathrm{B}$ to $96\,\mathrm{GiB}$ ($192\,\mathrm{GiB}$ bottom).}
  \label{fig:sizebench}
\vskip-20pt
\end{figure}

\subsection{Results} \label{sec:results}
For small arrays up to around twice the L2~cache size, the \emph{generic}
implementation is the clear winner.  As cache misses are mostly absent at
these array sizes, the overhead of Radsort cannot be outweighed by its
cache-locality benefits.
As input size exceeds this limit, Radsort and its variants quickly outperform
the generic implementation.  Implementing bit-manipulation based end-of-block
checking yields a speedup of~5--50\,\% on all systems tested (including others
not shown here), except for the \emph{power} system where it is slightly slower.

As expected, the \emph{swc}~variant performs very
consistently over the whole range of array lengths, as it minimises cache effects
through manual write combining.  This consistency is maintained even as the array
size approaches a significant chunk of the main memory, while the generic
implementation drops to around half of its performance, matching the results
observed by Wassenberg~et~al.~\cite{Wassenberg2010}.
The \emph{bitmanip} variant exceeds this performance on all
systems for array sizes above the L2~cache size, showing that its timely reuse
of input blocks for output blocks avoid read-to-own transfers just as effectively
as \emph{swc}.

Whenever it is worth using Radsort over an out-of-place radix sort, the parallel
variant outperforms the scalar variants, though the advantages quickly diminish
(cf.~Fig.~\ref{fig:threadbench}), topping out around 24~threads on the
\emph{icelake} machine and around 48~threads on the \emph{power} machine.
NUMA heavily impacts the \emph{grace} machine, causing high measurement
variance based on thread placement. Performance peaks here at around 72~threads.

In summary, it seems advantageous to use \emph{generic} for short
arrays, and to switch to \emph{bitmanip} or \emph{parallel} as input exceeds some
empirically determined threshold.
This threshold could be set to the size of~$T$, allowing for the reuse of~$T$ as
the output array~$A'$ of~Alg.~\ref{alg:ooplsd}.

\section{Related Work}
Radsort is an attempt to adapt the ideas of IPS${}^2$RA~\cite{Axtmann2022}
to LSD radix sort.  While IPS${}^2$RA is designed to use
$\mathcal O(1)$~extra memory, the author found this difficult
to achieve under the LSD~approach's stability requirements.
By raising the overhead to~$\mathcal O(\sqrt n)$, the permutation of blocks can
be explicitly tracked in~$\pi$, addressing these challenges.
As a consequence, blocks only have to actually be moved around during finalisation,
giving a significant speedup.

The result is similar to the ``out-of-cache, in-place, list of blocks'' partitioning
described by Polychroniou et al.~\cite[\S~3.2.3]{Polychroniou2014}, but uses arrays
over linked lists to track~$\pi$, improving cache utilisation, and permitting a
straightforward finalisation procedure.

A major performance limitation of radix sorts is the need for
output buckets to be read into cache so that they can be written to (read for
ownership), effectively halving the write bandwidth and causing cache misses.
While CPUs provide write-combining buffers to avoid this problem, radix sorts
have a fan-out that exceeds their number, rendering them ineffective.
The performance impact can be reduced by prefetching the output buckets
speculatively~\cite{Downs2019}, or more thoroughly by manual software-defined
write-combining~\cite{Wassenberg2010}.

Radsort solves the problem more elegantly: as input blocks are
reused for output blocks after only a short delay, the output blocks are still
hot in cache\footnote{As long as the cache fits up to $\sigma b$~elements
that have been consumed but not yet reused as output blocks.}
when written to, avoiding an extra read for ownership.
Consequently, we have found prefetching to have little to no performance impact
and software write-combining to not be needed.

\section{Future Work}
It is promising to adapt Radsort to the sorting of dissimilarly sized
elements, such as lines of text or JSON records, directly, i.\,e.~without resorting
to sorting an array of pointers to variable-length records.  Likewise, adapting
the general approach to external sorting should be investigated.

During the research for this paper, the authors investigated whether
Radsort's data structures would work for an MSD (i.\,e.~recursive) radix sort
procedure.
Keeping the data structures identical, we quickly run into problems: while each
recursive iteration consumes one partial block from the bucket it recurses over,
it produces up to $\sigma$~new partial blocks, leading to an uncompetitive memory
overhead.
The authors found that this problem can be addressed by compacting the partial
blocks into a sequence of dense blocks after each iteration, tracking the beginning
and end of partial data within this sequence using an extra data structure.
This reduces the extra space for partial blocks to one block per recursion level plus
$2\sigma$~words to track the number of elements in the partial tail of each bucket,
requiring $\mathcal O\bigl((\sigma+b)n_t\bigr)$ extra space.
Due to time constraints, this idea was not explored further.

\section{Conclusion}
In Radsort, we provide a straightforward way to implement LSD radix sorting with
just $\mathcal O(\sqrt n)$~space overhead.  Radsort readily parallelises to a moderate
number of threads.
The performance can be further improved using bit manipulation techniques,
at the cost of a more complex implementation.

As a consequence of more effective cache utilisation,
the performance of Radsort is competitive with a standard out-of-place LSD radix
sort for arrays as small as~$2\,\mathrm{MiB}$, even when using only one thread.
The read-to-own bottleneck on large datasets is avoided without the requirement of
platform-specific techniques like software-defined write-combining.

\bibliography{radsort}
\end{document}